\documentclass[conference]{IEEEtran}
\usepackage{graphicx}

\usepackage{bm}
\usepackage{amsfonts}
\usepackage{amsthm}
\usepackage{amssymb}
\usepackage{color}
\usepackage{mathrsfs}
\usepackage{upgreek}

\usepackage{hyperref}
\usepackage{pstricks}

\usepackage{cite}

\usepackage{booktabs}
\usepackage{balance}

\newtheorem{thm}{Theorem}
\newtheorem{lemma}{Lemma}

\newtheorem{proposition}{Proposition}

\ifCLASSINFOpdf
\else
\fi
\usepackage{amsmath}
\usepackage{amsthm}

\usepackage{algorithm}
\usepackage{algorithmic}
\ifCLASSOPTIONcompsoc
 \usepackage[caption=false,font=normalsize,labelfont=sf,textfont=sf]{subfig}
\else
 \usepackage[caption=false,font=footnotesize]{subfig}
\fi
\begin{document}
%

\title{Deterministic-Random Tradeoff of Integrated Sensing and Communications in Gaussian Channels: \\ A Rate-Distortion Perspective}
%
%
%

\author{\IEEEauthorblockN{Fan Liu\IEEEauthorrefmark{1},
Yifeng Xiong\IEEEauthorrefmark{2},
Kai Wan\IEEEauthorrefmark{3},
Tony Xiao Han\IEEEauthorrefmark{4},
and Giuseppe Caire\IEEEauthorrefmark{5}}
\IEEEauthorblockA{\IEEEauthorrefmark{1}Southern University of Science and Technology, China}
\IEEEauthorblockA{\IEEEauthorrefmark{2}Beijing University of Posts and Telecommunications, China}
\IEEEauthorblockA{\IEEEauthorrefmark{3}Huazhong University of Science and Technology, China}
\IEEEauthorblockA{\IEEEauthorrefmark{4}Huawei Technologies Co. Ltd., China}
\IEEEauthorblockA{\IEEEauthorrefmark{5}	Technische Universit\"at Berlin, Germany}
}

\maketitle

\begin{abstract}
Integrated sensing and communications (ISAC) is recognized as a key enabling technology for future wireless networks. To shed light on the fundamental performance limits of ISAC systems, this paper studies the deterministic-random tradeoff between sensing and communications (S\&C) from a rate-distortion perspective under vector Gaussian channels. We model the ISAC signal as a random matrix that carries information, whose realization is perfectly known to the sensing receiver, but is unknown to the communication receiver. We characterize the sensing mutual information conditioned on the random ISAC signal, and show that it provides a universal lower bound for distortion metrics of sensing. Furthermore, we prove that the distortion lower bound is minimized if the sample covariance matrix of the ISAC signal is deterministic. We then offer our understanding of the main results by interpreting wireless sensing as non-cooperative source-channel coding, and reveal the deterministic-random tradeoff of S\&C for ISAC systems. Finally, we provide sufficient conditions for the achievability of the distortion bound by analyzing a specific example of target response matrix estimation.
\end{abstract}


%
\IEEEpeerreviewmaketitle

\section{Introduction}
\subsection{Background and Motivation}
The future wireless networks are anticipated to possess high-precision and robust wireless sensing capability in addition to the communications functionality, thus to support a variety of emerging applications, ranging from intelligent transportation to smart cities and homes. To that end, ISAC system, which enables the shared use of hardware, spectrum, and signaling resources between S\&C, are envisioned to be one of the game-changing techniques for 6G and Wi-Fi-7 networks \cite{dual_functional}. While ISAC signal processing and waveform design have been extensively studied in the past few years, its fundamental limits and the resulting performance tradeoff between S\&C were less understood, which have been long-standing open in the research community \cite{anliu_fundamantal}.

\subsection{Existing Works}
Both S\&C focus on processing signals and information, whose theoretical foundations are built upon estimation, detection, and information theories. In particular, sensing is to extract useful information about targets of interest from observed echo signals, whereas communication is to recover the information encoded by the transmitter from the received signals. For decades, S\&C are regarded as two separated research fields, despite that they are closely related to each other as an ``information-theoretic odd couple" \cite{odd_couple}. Indeed, the fundamental theories and performance metrics of S\&C may be bridged in a variety of ways. The most well-known result is the I-MMSE equation \cite{immse}, which states that for a scalar Gaussian channel, the derivative of the mutual information (MI) between the input and output with respect to the signal-to-noise ratio (SNR) is the minimum mean squared error (MMSE) of estimating the input from the output. The I-MMSE equation may also be deduced from the De Bruijn identity \cite{immse}, which connects the differential entropy and Fisher information. Moreover, the detection probability and relative entropy can be linked to each other via the celebrated Stein's lemma \cite{MIMO_radar_KLD}.

To reveal more insights into the fundamental limits of ISAC systems, recent works modeled the monostatic radar sensing as a delayed feedback channel depending on the target states \cite{cdit}. In such a case, ISAC transmission can be treated as a joint state estimation and communication problem, whose performance limits are characterized by the tradeoff between the communication capacity and state estimation distortion. More relevant to this work, the performance limits of ISAC were depicted by the Cram\'er-Rao bound (CRB)-communication rate
 region in \cite{xiong2022flowing,xiong2022fundamental}, which unveiled the fundamental deterministic-random tradeoff (DRT) between S\&C.

\subsection{Contributions of This Paper}
In this paper, we re-examine the ISAC performance tradeoff from a rate-distortion perspective. In particular, we extend the DRT between the communication rate and CRB in \cite{xiong2022flowing,xiong2022fundamental} to any well-defined distortion metrics for sensing. We commence by modeling the ISAC signal emitted from an ISAC transmitter (Tx) as a random signal carrying information intended for the communication receiver (Rx), which is perfectly known to the sensing Rx as a reference waveform as in typical radar applications. By analyzing the sensing MI conditioned on the randomly varying ISAC signal, we show that it provides a lower bound for well-defined sensing distortion metrics, e.g., MSE and detection probability. We then underline the generic DRT in an ISAC system by proving that the distortion lower bound is minimized if the sample covariance matrix of the ISAC signal is deterministic, in which case the achievable communication rate is reduced owing to the decrease of randomness, or equivalently, the reduced degrees of freedom (DoFs) in the signal. As a step further, we provide a discussion on the main results by offering a new angle that interprets the sensing operation in an ISAC system as a non-cooperative source-channel coding system, where the target as a non-cooperative source transmits the information about its parameters to the sensing Rx in a passive manner. Finally, we analyze an example of target response matrix estimation, and provide sufficient conditions for the achievability of the proposed lower bound.

\begin{figure}[t]
\centering
\begin{minipage}{.2\textwidth}
 \centering
\includegraphics[width=.99\textwidth]{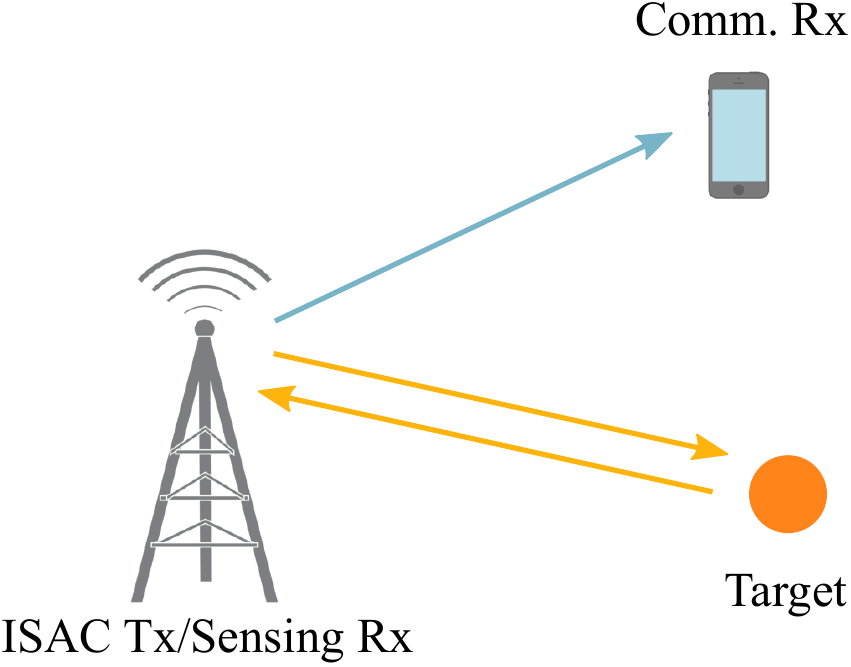}
\vspace{-1mm}
\footnotesize  (a) Monostatic sensing
\end{minipage}
\hspace{5mm}
\begin{minipage}{.21\textwidth}
\centering
\includegraphics[width=.99\textwidth]{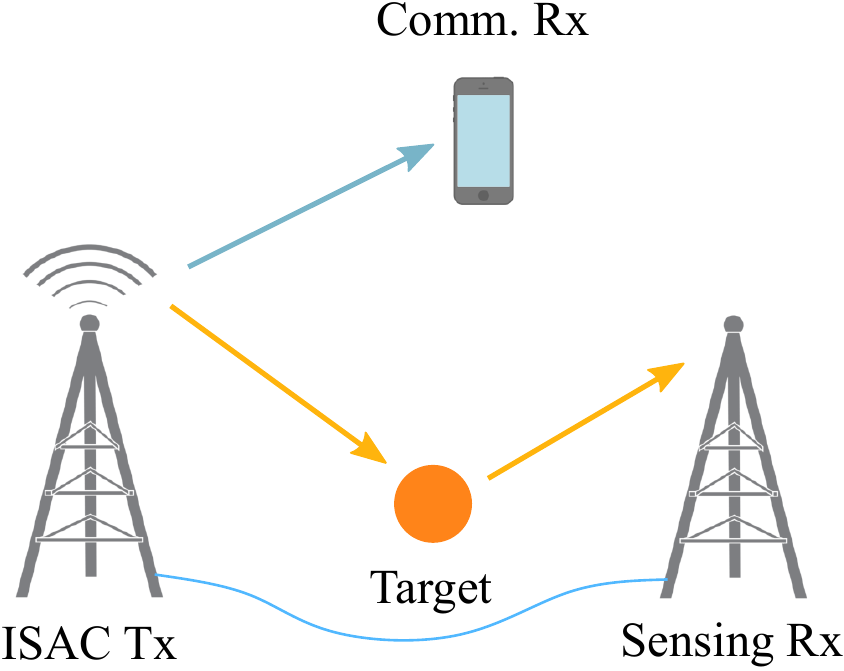}
\footnotesize  (b) Bistatic sensing
\end{minipage}
\caption{The ISAC scenarios considered in this paper, where the dual-functional waveform $\mathbf{X}$ is known to both the ISAC transmitter and sensing receiver.}
\label{fig:scenarios}
\end{figure}

\section{System Model and Performance Metrics}
\subsection{System Model}
We consider a generic point-to-point (P2P) downlink ISAC system consisting of an ISAC Tx, a sensing Rx, a communication Rx, and one or multiple targets. The ISAC Tx sends a dual-functional signal to perform target sensing and downlink communication simultaneously. Under a vector Gaussian channel, the S\&C signals received at the sensing Rx and communication Rx can be respectively modeled as
\begin{equation}\label{eq1}
\begin{gathered}
  {{\mathbf{Y}}_s} = {{\mathbf{H}}_s}\left( {\boldsymbol{{\upeta}}} \right){\mathbf{X}} + {{\mathbf{Z}}_s}, \hfill \\
   {{\mathbf{Y}}_c} = {{\mathbf{H}}_c}{\mathbf{X}} + {{\mathbf{Z}}_c},\hfill \\ 
\end{gathered}
\end{equation}
where $\mathbf{X} \in \mathbb{C}^{M \times T}$ is the dual-functional signal matrix transmitted from the ISAC Tx, with $M$ being the number of antennas at the ISAC Tx and $T$ being the number of discrete samples; ${{\mathbf{H}}_s} \in \mathbb{C}^{N_s \times M}$ and ${{\mathbf{H}}_c} \in \mathbb{C}^{N_c \times M}$ are S\&C channel matrices, with $N_s$ and $N_c$ being the numbers of antennas at the sensing and communication Rxs, respectively; ${{\mathbf{Z}}_s} \in \mathbb{C}^{N_s \times T}$ and ${{\mathbf{Z}}_c} \in \mathbb{C}^{N_c \times T}$ are zero mean white Gaussian noise matrices with variances $\sigma_s^2$ and $\sigma_c^2$, respectively. In particular, $\boldsymbol{{\upeta}} \in \mathbb{R}^{K}$ represents the target parameters of interest, e.g., angle, range, and velocity, with $K$ being the dimension of the target parameters. The sensing channel ${{\mathbf{H}}_s}:\mathbb{R}^{K} \to \mathbb{C}^{N_s \times M}$ is assumed to be a deterministic function of $\boldsymbol{{\upeta}}$. 
 Since the sensing Rx is typically collocated with the ISAC Tx (monostatic sensing), or is connected with the ISAC Tx with an optical fiber (bistatic sensing), the dual-functional signal $\mathbf{X}$ is known to both ISAC Tx and sensing Rx, which is valid for most of radar applications. On the other hand, as $\mathbf{X}$ contains useful information intended for the communication Rx, it is unknown to the communication Rx. Therefore, we model $\mathbf{X}$ as a random matrix following a distribution ${{p_{\mathbf{X}}}\left( {\bm{X}} \right)}$, whose realization is known to the ISAC Tx and sensing Rx, but is unknown to the communication Rx. We also assume $\mathbb{E}\left\{\mathbf{X}\right\} = \mathbf{0}$, and denote the sample and statistical covariance matrices of $\mathbf{X}$ as $\mathbf{R}_X =  T^{-1} \mathbf{X}\mathbf{X}^H$, and ${\bm{\tilde R}_X} = \mathbb{E}\left\{\mathbf{R}_X\right\}$, respectively.


Accordingly, we define S\&C tasks in the ISAC system as
\begin{itemize}
\item \textbf{Sensing Task:} Estimate $\boldsymbol{{\upeta}} \in \mathbb{R}^{K}$ from the observation $\mathbf{Y}_s$ at the sensing Rx, with the knowledge of the probing signal $\mathbf{X}$.
\item \textbf{Communication Task:} Recover the useful information contained in $\mathbf{X}$ from the received signal $\mathbf{Y}_c$ at the communication Rx, with the knowledge (or statistical knowledge) of the channel ${{\mathbf{H}}_c}$.
\end{itemize}

Without loss of generality, we assume ${\boldsymbol{{\upeta}}} \sim p_{\boldsymbol{{\upeta}}}\left(\bm{\eta}\right)$, which vary every $T$ samples in an i.i.d. manner. Moreover, the communication channel is assumed to vary in an i.i.d. manner every $kT$ samples, with $k \in \mathbb{Z}^{+}$.

\subsection{Sensing Performance Metrics}
The sensing performance can be measured by various metrics that characterize either the estimation accuracy or detection reliability, namely, mean squared error (MSE), detection probability, and Cram\'er-Rao bound (CRB). These metrics (or their functions) may be induced from the distortion function $d\left(\boldsymbol{\upeta},\boldsymbol{\hat{\upeta}}\right)$  between the parameter $\boldsymbol{\upeta}$ and its estimate $\boldsymbol{\hat{\upeta}}$ in the context of rate-distortion theory \cite{shannon1959coding,berger1971rate,nit}. For the estimation problem, the squared Euclidean distance $d\left(\boldsymbol{\upeta},\boldsymbol{\hat{\upeta}}\right) = \left\| {\boldsymbol{\upeta}-\boldsymbol{\hat{\upeta}}} \right\|^2$ is a commonly used distortion function, which induces the MSE metric. For the detection problem, $\upeta \in \left\{0,1\right\}$ is a random binary variable indicating whether the target is present or absent. To that end, one may choose the distortion as the Hamming distance between $\upeta$ and $\hat{\upeta}$, namely, $d\left(\upeta, \hat{\upeta}\right) = \upeta\oplus\hat{\upeta}$. The average distortion is then given as
\begin{equation}\label{eq_distortion}
\begin{gathered}
  \mathbb{E}\left\{ \upeta\oplus\hat{\upeta} \right\} =  \hfill \\  \left( {1 \oplus 1} \right)\Pr \left( {\hat{\upeta}  = 1\left| {\upeta  = 1} \right.} \right)  + \left( {0 \oplus 0} \right)\Pr \left( {\hat{\upeta}  = 0\left| {\upeta  = 0} \right.} \right) \hfill \\
   + \left( {1 \oplus 0} \right)\Pr \left( {\hat{\upeta}  = 1\left| {\upeta  = 0} \right.} \right)  + \left( {0 \oplus 1} \right)\Pr \left( {\hat{\upeta}  = 0\left| {\upeta  = 1} \right.} \right) \hfill \\
= 1 - {P_D} + {P_{FA}}, \hfill \\ 
\end{gathered}
\end{equation}
where $P_D$ and $P_{FA}$ stand for the detection and false-alarm probabilities, respectively. Under the Neyman-Pearson criterion where $P_{FA}$ is fixed, minimizing the average Hamming distortion (\ref{eq_distortion}) yields the maximum $P_D$. As we will show later, these distortion metrics connect closely to the sensing MI. Therefore, we will focus on characterizing the sensing MI and revealing its connection with other sensing metrics.

For the sensing model in (\ref{eq1}), the sensing MI is defined as
\begin{equation}\label{eq4}
{I_s} = I\left( {{{\mathbf{Y}}_s};\boldsymbol{\upeta} \left| {\mathbf{X}} \right.} \right).
\end{equation}
At the first glance, the MI (\ref{eq4}) is unlikely to be simplified due to the possible nonlinear dependence between ${{\mathbf{Y}}_s}$ and $\boldsymbol{\upeta}$, namely, the nonlinearity of ${{\mathbf{H}}_s}\left( {\boldsymbol{{\upeta}}} \right)$. Fortunately, the following lemma admits a more tractable form of the sensing MI.
\begin{lemma}
The MI between ${{{\mathbf{Y}}}_s}$ and $\boldsymbol{\upeta}$ equals to that between ${{{\mathbf{Y}}}_s}$ and ${{{\mathbf{H}}}_s}$, i.e.,
\begin{equation}\label{eq5}
I_s = I\left( {{{{\mathbf{Y}}}_s};\boldsymbol{\upeta}\left| {\mathbf{X}} \right.} \right) = I\left( {{{{\mathbf{Y}}}_s};{{{\mathbf{H}}}_s}\left| {\mathbf{X}} \right.} \right).
\end{equation}
\end{lemma}
\renewcommand{\qedsymbol}{$\blacksquare$}
\begin{proof}
See Appendix A.
\end{proof}

\subsection{Communication Performance Metric}
The communication performance can be measured by the ergodic achievable rate, which is expressed as
\begin{equation}\label{eq2}
{I_c} = \mathop {\max }\limits_{{p_{\mathbf{X}}}\left( {\bm{X}} \right)} \;T^{-1}I\left( {{{\mathbf{Y}}_c};{\mathbf{X}}\left| {{{\mathbf{H}}_c}} \right.} \right),\;\;\operatorname{s.t.}\;\;{p_{\mathbf{X}}}\left( \bm{X} \right) \in \mathcal{F},
\end{equation}
where $I\left( {{{\mathbf{Y}}_c};{\mathbf{X}}\left| {{{\mathbf{H}}_c}} \right.} \right)$ stands for the mutual information (MI) between $\mathbf{Y}_c$ and $\mathbf{X}$ conditioned on $\mathbf{H}_c$, and $\mathcal{F}$ represents the feasible set of the distribution ${{p_{\mathbf{X}}}\left( {\bm{X}} \right)} $ under some constraint such as power and sensing performance constraints.

\section{Main Results}
With Lemma 1 at hand, we first prove that the sensing MI has the following property.
\begin{lemma}
$I\left( {{{\mathbf{Y}}_s};\boldsymbol{\upeta} \left| {\mathbf{X}} \right. = \bm{A}} \right)$ is a concave function in $\bm{R}_A =  T^{-1} \bm{A}\bm{A}^H$.
\end{lemma}
\begin{proof}
See Appendix B.
\end{proof}
With Lemma 2, we may write $I\left( {{{\mathbf{Y}}_s};\boldsymbol{\upeta} \left| {\mathbf{X}} \right. = \bm{A}} \right)$ as a function of ${\bm{R}}_A$, namely, ${I_{\boldsymbol{{\upeta}}} }\left( {{{\bm{R}}_A}} \right)$. We then prove that the following proposition holds.
\begin{proposition}
Let the average transmit power be $P_T$, namely, $\operatorname{tr}\left(\mathbb{E}\left\{\mathbf{R}_X \right\}\right) = P_T$. The sensing MI (\ref{eq4}) is maximized if and only if the support of the sample covariance matrix $\mathbf{R}_X =  T^{-1} \mathbf{X}\mathbf{X}^H$ is the solution set of the following deterministic convex optimization problem
\begin{equation}\label{eq10}
\begin{gathered}
  \mathop {\max }\limits_{{{\bm{R}}_A} \succeq {\mathbf{0}},\;{{\bm{R}}_A} = {\bm{R}}_A^H} \;{I_{\boldsymbol{{\upeta}}} }\left( {{{\bm{R}}_A}} \right)\quad\operatorname{s .t.}\;\;\;\operatorname{tr} \left( {{{\bm{R}}_A}} \right) = {P_T}, \hfill \\ 
\end{gathered}
\end{equation}
in which case $\mathbf{R}_X$ has a deterministic trace. In particular, if problem (\ref{eq10}) has a unique solution, then $\mathbf{R}_X$ itself is deterministic, i.e., $\mathbf{R}_X = \mathbb{E}\left\{\mathbf{R}_X\right\} = \bm{\tilde R}_X$. 
\end{proposition}
\begin{proof}
Since the objective function is a concave function in ${\bm{R}}_A$, the proof is a straightforward modification of \cite[Proposition 3]{xiong2022fundamental}, which is omitted here for brevity.
\end{proof}
We then show that the following theorem provides a universal bound for the distortion of recovering $\boldsymbol{\upeta}$ from $\mathbf{Y}_s$.
\begin{thm}
(Distortion Lower Bound) Let $D\left(R\right)$ be the distortion-rate function for the to-be-sensed i.i.d. random parameter $\boldsymbol{\upeta}\sim p_{\boldsymbol{{\upeta}}}\left(\bm{\eta}\right)$, $\hat{\boldsymbol{\upeta}}$ an estimate of $\boldsymbol{\upeta}$, and $d\left(\boldsymbol{\upeta},\hat{\boldsymbol{\upeta}}\right)$ the corresponding distortion function measuring the sensing performance. The average distortion of recovering $\boldsymbol{\upeta}$ from the noisy observation $\mathbf{Y}_s$ is lower-bounded by 
\begin{equation}\label{eq11}
\mathbb{E}\left\{ {d\left(\boldsymbol{\upeta},\hat{\boldsymbol{\upeta}}\right)} \right\} \mathop  \ge\limits^{\left( a \right)} D\left[ {\mathbb{E}\left\{ {{I_{\boldsymbol{\upeta}} }\left( {{{\mathbf{R}}_X}} \right)} \right\}} \right] \mathop  \ge\limits^{\left( b \right)} D\left( { {{I_{\boldsymbol{\upeta}} }\left( {{{{\bm{\tilde R}}}_X}} \right)} } \right),
\end{equation}
where the equality holds for $\left(b\right)$ if $\mathbf{R}_X$ satisfies the conditions in Proposition 1. In particular, if the solution of (\ref{eq10}) is unique, then $\mathbf{R}_X$ itself should be deterministic to achieve $\left(b\right)$.
\end{thm}
\begin{proof}
See Appendix C.
\end{proof}
When the sensing-optimal sample covariance matrix is unique, i.e., when (\ref{eq10}) has a unique solution $\bm{R}_s^\star$, the communication capacity will be reduced due to the extra constraint, or, equivalently, the loss of DoFs in the ISAC signal. Theorem 2 provides the high-SNR asymptotic egordic communication capacity under a fixed sample covariance matrix.
\begin{thm}
(Sensing-Limited High-SNR Ergodic Capacity) Suppose that problem (\ref{eq10}) has a unique solution ${{\bm{R}}_s^ \star } $. In the high-SNR regime, namely, when $P_T/\sigma_c^2\to \infty$, the rate $I_c$ can be expressed as
\begin{equation}\label{high_SNR_capacity}
\begin{gathered}
  {I_c} = \mathop {\max }\limits_{{p_{\mathbf{X}}}\left( \bm{X} \right)} \;{T^{ - 1}}I\left( {{{\mathbf{Y}}_c};{\mathbf{X}}\left| {{{\mathbf{H}}_c}} \right.} \right),\;\;\operatorname{s.t.}\;\;{T^{ - 1}}{\mathbf{X}}{{\mathbf{X}}^H} = {{\bm{R}}_s^ \star } \hfill \\
  \;\;\;\; = \mathbb{E}\left\{ {\left( {1 - \frac{L}{{2T}}} \right)\log \left| {\sigma _c^{ - 2}{{\mathbf{H}}_c}{{\bm{R}}_s^ \star }{\mathbf{H}}_c^H} \right| + {c_0}} \right\} + O\left( {\sigma _c^2} \right), \hfill \\ 
\end{gathered}
\end{equation}
where $L = \operatorname{rank}\left({{\mathbf{H}}_c}{{\bm{R}}_s^ \star }{\mathbf{H}}_c^H\right)$, and the term
\begin{equation}
{c_0} = \frac{L}{T}\left[ {\left( {T - \frac{L}{2}} \right)\log \frac{T}{e} - \log \Gamma \left( T \right) + \log 2\sqrt \pi  } \right]
\end{equation}
converges to zero as $T\to\infty$, where $\Gamma\left(\cdot\right)$ is the Gamma function.
\end{thm}
\begin{proof}
See the proof of \cite[Theorem 1]{xiong2022fundamental}.
\end{proof}
We next discuss the implications of the above main results.

\section{Discussions}
\subsection{Wireless Sensing as Non-Cooperative Joint Source-Channel Coding}
It is well-known that the communication MI $I\left( {{{\mathbf{Y}}_c};{\mathbf{X}} \left| {\mathbf{H}}_c \right.} \right)$ has an explicit operational meaning, i.e., achievable data rate. Nevertheless, it remains unclear what the operational meaning is for the sensing MI $I\left( {{{\mathbf{Y}}_s};\boldsymbol{\upeta} \left| {\mathbf{X}} \right.} \right)$, as it does not seem to indicate any ``coding" rate in wireless sensing systems.

The main results in Sec. III offer an interesting angle to look at the above issue by interpreting wireless sensing from an information-theoretic viewpoint. That is, one may consider wireless sensing as a procedure of {\textit{non-cooperative joint source-channel coding}}, where the target encodes the information of the parameter $\boldsymbol{\upeta}$, and communicates it to the sensing Rx in a passive and non-cooperative manner. In such a case, the random parameter $\boldsymbol{\upeta}$ is regarded as a memoryless source, $\mathbf{H}_s\left(\boldsymbol{\upeta}\right)$, the channel input, is a letter of the channel codeword, and $\mathbf{Y}_s$ is the channel output, from which a distorted version of $\boldsymbol{\upeta}$ (i.e., $\hat{\boldsymbol{\upeta}}$) may be revealed at the sensing Rx. The rate-distortion function $R\left(D\right)$ characterizes the minimum number of bits required to communicate $\boldsymbol{\upeta}$ to the sensing Rx at an allowable distortion $D$.

More interestingly, the ISAC signal $\mathbf{X}$ serves as the ``channel matrix" for communicating the source $\boldsymbol{\upeta}$. In particular, as we model $\mathbf{X}$ as a random variable, it can be treated as an ``ergodic channel" with channel state information at the sensing Rx (CSIR), in the sense that it varies in an i.i.d. manner every $T$ symbols, and is perfectly known to the sensing Rx. Accordingly, the sensing MI $I\left( {{{\mathbf{Y}}_s};\boldsymbol{\upeta} \left| {\mathbf{X}} \right.} \right)$ becomes an ``ergodic rate" of the channel, in the sense that it bounds the rate-distortion function $R\left(D\right)$, and thereby yields a lower-bound for the average distortion of any well-defined distortion functions. Indeed, the proof of Theorem 1 is similar to the proof of converse for the source-channel separation theorem under ergodic channels, where $R\left(D\right)$ is only determined by the distribution of the target (the source $p_{\boldsymbol{{\upeta}}}\left(\bm{\eta}\right)$), and the maximum sensing MI relies solely on the statistical model (the channel ${p_{\left. {\mathbf{Y}}_s \right|\mathbf{H}_s }}\left( {{\bm{Y}_s}\left| {\bm{H}_s}  \right.} \right)$) \cite{IT_Polyanskiy_Wu}. Nevertheless, it is unlikely to prove the general achievability of such a bound. This is because the separation theorem requires a {\textit{block-wise}} coding strategy, which encodes a sequence of i.i.d. $\boldsymbol{\upeta}$, i.e., $\boldsymbol{\upeta}^k = \left(\boldsymbol{\upeta}_1, \boldsymbol{\upeta}_2, \ldots, \boldsymbol{\upeta}_k\right)$ by a channel codeword $\mathbf{H}_s^n = \left(\mathbf{H}_{s,1}, \mathbf{H}_{s,2}, \ldots, \mathbf{H}_{s,n}\right)$. In our case, the target as a non-cooperative information source does not have such a block-wise ``coding" capability. The only thing it can do is to map each source letter $\boldsymbol{\upeta}$ to a channel code letter $\mathbf{H}_s\left(\boldsymbol{\upeta}\right)$, which is essentially a {\textit{letter-wise}} coding strategy that does not possess the general optimality.
\begin{figure}[!t]
    \centering
    \includegraphics[width=\columnwidth]{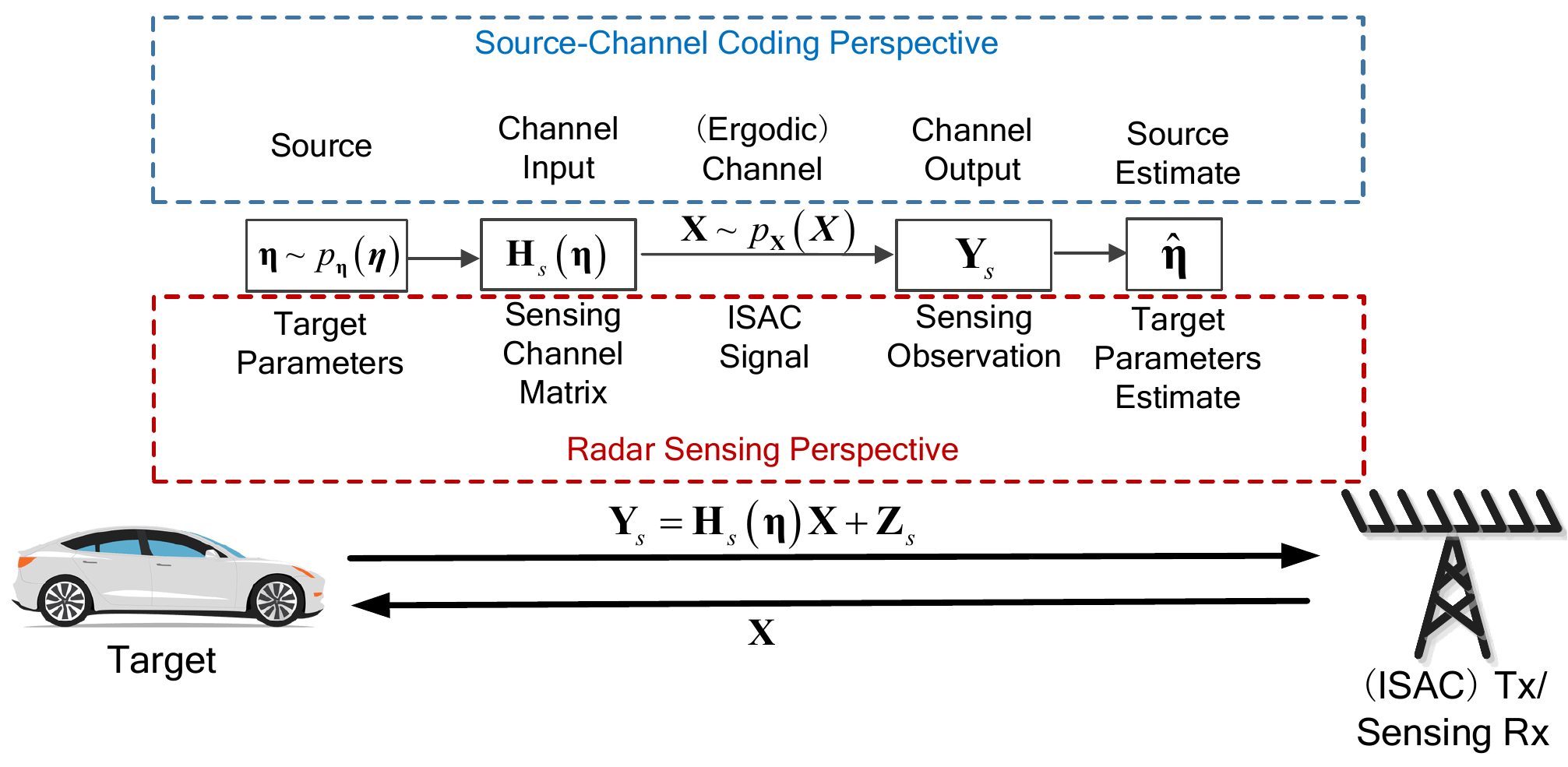}
    \caption{Wireless sensing as non-cooperative joint source-channel coding.}
    \label{fig:2}
\end{figure}

\subsection{Deterministic-Random Tradeoff in ISAC Systems}
Since the establishment of the Shannon theory, it has been well confirmed by both the academia and engineering practice that communication signals should be ``as random as possible" to convey information. In contrast to that, sensing (radar) systems favor deterministic signals to achieve a stable estimation/detection performance. One example is that in order to maximize the SINR of the target return, radar typically emits high-power constant-modulus signals to overcome the nonlinear amplifier distortion, where only the signal phases are allowed to vary. 

The lessons learned from Theorems 1 are that, ISAC signals should be deterministic to an extent to reach the optimal sensing MI (and thereby the optimal distortion lower bound), in the sense that the support of the sample covariance matrix $\mathbf{R}_X$ should be restricted to the optimal solution set of (\ref{eq10}), or even be deterministic itself. Under such a constraint, as indicated by Theorem 2, the resulting achievable communication rate is strictly less than the Gaussian capacity due to the loss of the pre-log DoFs. In fact, it was recently proved that the high-SNR capacity (\ref{high_SNR_capacity}) is attained by the uniform distribution over Stefiel manifold, which is no longer Gaussian signaling anymore \cite{xiong2022flowing}. In other words, the ISAC system trades off the randomness in the signal (and thereby the communication performance) for achieving a better sensing performance, which is referred to as deterministic-random tradeoff (DRT) between S\&C functionalities in the ISAC system \cite{xiong2022fundamental}. 

\begin{figure}[!t]
    \centering
    \includegraphics[width=0.6\columnwidth]{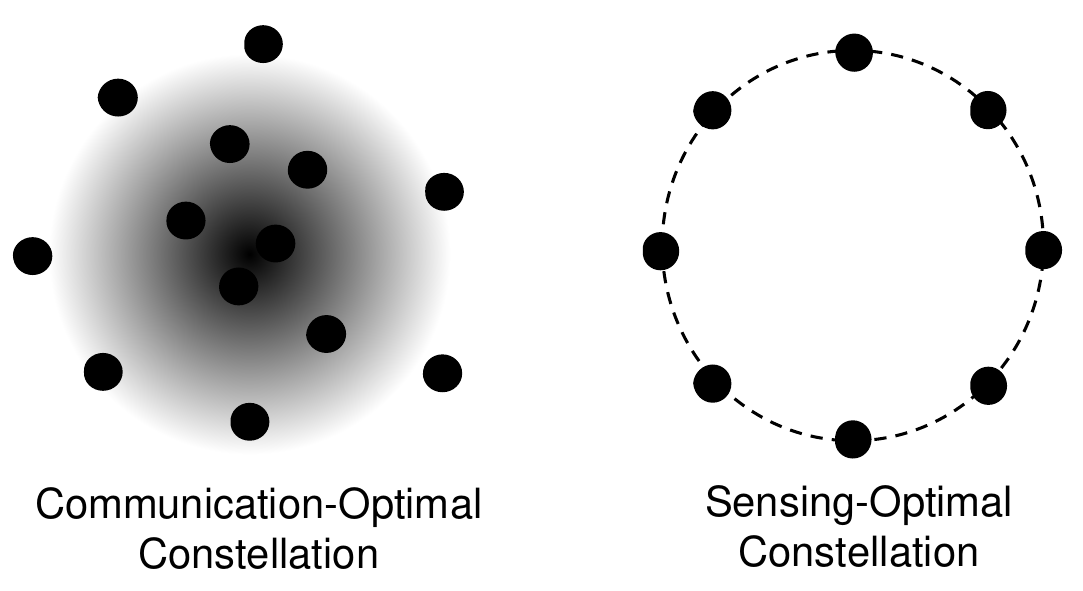}
    \caption{Determinitstic-Random Tradeoff in ISAC systems for scalar signals.}
    \label{fig:3}
\end{figure}
More remarkably, the lower bound in (\ref{eq11}) does not require a specific distortion function for sensing, which implies that the DRT may hold for a variety of sensing metrics, including the MSE, detection probability, and negative sensing MI itself. Note that the proof does not need any analytical expressions of distortion metrics, which is a generalization of the results in \cite{xiong2022fundamental} relying on the convexity of the closed-form expression of CRB. Despite the fact that the achievability conditions of (\ref{eq11}) remain unexplored, Theorem 1 provides strong evidence for the generic correctness of the DRT in ISAC systems.


\section{Examples}
In this section, we analyze an example of target response matrix (TRM) estimation, where we have $\boldsymbol{\upeta} = \operatorname{vec}\left({\mathbf{H}_s}\right) = \mathbf{h}_s$. To explicitly characterize the sensing MI, we also assume that the target parameters are zero-mean circular-symmetric Gaussian distributed with an invertible statistical covariance matrix $\bm{\Tilde{R}}_h$, i.e., $\mathbf{h}_s \sim \mathcal{CN}\left(\mathbf{0}, \bm{\Tilde{R}}_h\right)$. Such a model has been widely applied to extended target estimation of MIMO radar systems. We will also show that, under certain conditions, the lower bound in (\ref{eq11}) is achievable.
\subsection{Scalar Case}
We commence by examining the scalar case, namely, $M = 1, T = 1$, where the S\&C signals reduce to
\begin{equation}\label{eq18}
\begin{gathered}
  {{{Y}}_s} = {{H}_s}{X} + {{Z}_s}, \hfill \\
  {Y_c} = {H_c}X + {Z_c}, \hfill \\ 
\end{gathered}
\end{equation}
where $H_s \sim \mathcal{CN}\left(0, \sigma^2_h\right)$. The sensing MI is given by
\begin{equation}\label{eq19}
\begin{gathered}
 I_s = I\left( {{Y_s};{H_s}\left|X \right.} \right) = \mathbb{E}\left\{ {\log \left( {1 + \frac{{{{\left| X\right|}^2}\sigma _h^2}}{{\sigma _s^2}}} \right)} \right\} \hfill \\
 \le \log \left( {1 + \frac{{\mathbb{E}\left\{ {{{\left|X\right|}^2}} \right\}\sigma _h^2}}{{\sigma _s^2}}} \right) = \log \left( {1 + \frac{{ P_T \sigma _h^2}}{{\sigma _s^2}}} \right) \triangleq I_{s,\max},\hfill \\ 
\end{gathered}
\end{equation}

Let us denote the estimate of $H_s$ as ${\hat H}_s$. By leveraging the squared Euclidean distance distortion, the sensing performance is measured by the MSE. The following proposition provides the achievability condition of the MSE lower bound.
\begin{proposition}
The scalar-case sensing MSE is bounded by
\begin{equation}
    \mathbb{E}\left\{ {{{\left| {{H_s} - {{\hat H}_s}} \right|}^2}} \right\} \ge \sigma _h^2{2^{ - {I_{s,\max }}}}. 
\end{equation}
The equality holds if and only if the MMSE estimator is employed at the sensing receiver, and that ${\left|X \right|^2} = \mathbb{E}\left\{{{{\left|X \right|}^2}} \right\} = {P_T}$, i.e., $X$ has a constant amplitude $\sqrt{P_T}$.
\end{proposition}
\begin{proof}
    See Appendix D.
\end{proof}
\textit{Remark 1:} When the above lower bound is achieved, the ISAC system may only employ PSK constellations to convey information as shown in Fig. 3, leading to an achievable communication rate strictly lower than that of a Gaussian distributed $X$.

\textit{Remark 2:} The sensing model in (\ref{eq18}) is nothing but an uncoded transmission scheme of the Gaussian source $H_s$. For non-fading Gaussian channels, uncoded transmission is known to be optimal with MSE as the distortion. Nevertheless, in the ISAC system, $X$ is an i.i.d. random variable known to the sensing Rx, which makes (\ref{eq18}) essentially an ergodic channel with CSIR. In such a case, the capacity-achieving scheme is \textit{coding across blocks}, which cannot be realized due to the fact that the source is a non-cooperative target, as discussed in Sec. IV. Therefore, uncoded transmission is no longer optimal, resulting in the inequality $\left(a\right)$ in (\ref{eq23}). When the amplitude of $X$ is fixed to $\sqrt{P_T}$, the ergodic rate $I\left( {{Y_s};{H_s}\left|X \right.} \right)$ is maximized and equal to that of the non-fading channel, which makes uncoded transmission optimal again.
\subsection{Vector Case}
By noting $\boldsymbol{\upeta} = \operatorname{vec}\left({\mathbf{H}_s}\right) = \mathbf{h}_s$, we vectorize the sensing signal model as
\begin{equation}\label{eq25}
{{\mathbf{y}}_s} = \operatorname{vec} \left( {{{\mathbf{Y}}_s}} \right) = \left( {{{\mathbf{X}}^T} \otimes {{\mathbf{I}}_{{N_s}}}} \right){{\mathbf{h}}_s} + {{\mathbf{z}}_s} \triangleq  \mathbf{\tilde{X}}{{\mathbf{h}}_s} + {{\mathbf{z}}_s}.
\end{equation}
The sensing MI can be expressed as
\begin{equation}\label{eq26}
\begin{gathered}
  I\left( {{{\mathbf{y}}_s};{{\mathbf{h}}_s}\left| {{\mathbf{\tilde X}}} \right.} \right) 
   = \mathbb{E}\left\{ {\log\left| {{\mathbf{I}} + \sigma _s^{ - 2}{\mathbf{\tilde X}}{{\bm{\tilde{R}}}_h}{{{\mathbf{\tilde X}}}^H}} \right|} \right\} \hfill \\
   = \mathbb{E}\left\{ {\log\left| {{\mathbf{I}} + \sigma _s^{ - 2}{{\mathbf{\Lambda }}_h}{{\mathbf{U}}^H}\left( {{{\mathbf{X}}^*}{{\mathbf{X}}^T} \otimes {{\mathbf{I}}_{{N_s}}}} \right){\mathbf{U}}} \right|} \right\} \hfill \\
   = \mathbb{E}\left\{ {\log \left| {{\mathbf{I}} + \sigma _s^{ - 2}{{\bm{\Lambda }}_h}{{\mathbf{U}}^H}{\mathbf{K}}\left( {{{\mathbf{I}}_{{N_s}}} \otimes {{\mathbf{X}}^*}{{\mathbf{X}}^T}} \right){{\mathbf{K}}^T}{\mathbf{U}}} \right|} \right\}, \hfill \\ 
\end{gathered}
\end{equation}
where ${{\bm{\tilde{R}}}_h} = {\mathbf{U}}{{\bm{\Lambda }}_h}\mathbf{U}^H$ is the eigenvalue decomposition of ${{\bm{\tilde{R}}}_h}$, and $\mathbf{K}$ is a real commutation matrix satisfying $\mathbf{K}\mathbf{K}^T = \mathbf{K}^T\mathbf{K} = \mathbf{I}$, such that ${\mathbf{K}}\left( {{{\mathbf{I}}_{{N_s}}} \otimes {{\mathbf{X}}^*}{{\mathbf{X}}^T}} \right){{\mathbf{K}}^T} = {{{\mathbf{X}}^*}{{\mathbf{X}}^T} \otimes {{\mathbf{I}}_{{N_s}}}}$. By letting $\mathbf{\tilde F} = {{\mathbf{U}}^H}{\mathbf{K}} = \left[ {{\mathbf{\tilde F}_1},{\mathbf{\tilde F}_2}, \ldots ,{\mathbf{\tilde F}_{{N_s}}}} \right]$, and $\mathbf{F}_i = {{\mathbf{\tilde F}}_i^{*}}$, (\ref{eq26}) can be recast as
\begin{equation}\label{eq27}
\begin{gathered}
  I\left( {{{\mathbf{Y}}_s};{{\mathbf{H}}_s}\left| {\mathbf{X}} \right.} \right) =  
  \mathbb{E}\left\{ {\log\left| {{\mathbf{I}} + \sigma _s^{ - 2}T{{\mathbf{\Lambda }}_h}\sum\nolimits_{i = 1}^{{N_s}} {{{\mathbf{\tilde F}}_i}{\mathbf{R}}_X^*{\mathbf{\tilde F}}_i^H} } \right|} \right\} \hfill \\ 
   = \mathbb{E}\left\{ {\log  \left|{{\mathbf{I}} + \sigma _s^{ - 2}T{{\mathbf{\Lambda }}_h}\sum\nolimits_{i = 1}^{{N_s}} {{{\mathbf{F}}_i}{{\mathbf{R}}_X}{\mathbf{F}}_i^H} } \right|} \right\}. \hfill \\
\end{gathered}
\end{equation}
\begin{proposition}
The vector-case sensing MSE is bounded by
\begin{equation}
 \mathbb{E}\left\{ {{{\left\| {{\mathbf{H}_s} - {\mathbf{\hat H}_s}} \right\|}_F^2}} \right\} \ge {D_{VG}}\left[ {I\left( {{{\mathbf{Y}}_s};{{\mathbf{H}}_s}\left| {\mathbf{X}} \right.} \right)} \right],
\end{equation}
where $D_{VG}\left(\cdot\right)$ is the distortion-rate function for the independent Gaussian vector ${{\mathbf{U}}^H}{{\mathbf{h}}_s} \sim \mathcal{CN}\left(\mathbf{0}, \bm{\Lambda}_h\right)$. The above lower-bound is attained if the following conditions hold:

1. The MMSE estimator is employed at the sensing receiver.

2. The sample covariance matrix $\mathbf{R}_X$ is deterministic, i.e., ${{\mathbf{R}}_X} = \mathbb{E}\left( {{{\mathbf{R}}_X}} \right) = {\bm{\tilde R}}_X$.

3. The sum ${\sum\nolimits_{i = 1}^{{N_s}} {{{\mathbf{F}}_i}{\bm{\tilde R}}_X{\mathbf{F}}_i^H} }$ is diagonalizable by some unitary matrix $\mathbf{V}$, namely, 
\begin{equation}\label{eq32}
        {{\mathbf{V}}^H}\left( {\sum\nolimits_{i = 1}^{{N_s}} {{{\mathbf{F}}_i}{\bm{\tilde R}}_X{\mathbf{F}}_i^H} } \right){\mathbf{V}} = \operatorname{diag} \left( {{\beta _1},{\beta _2}, \ldots ,{\beta _{{N_s}M}}} \right).
    \end{equation}
    In particular,  $\beta_i$ should have the following water-filling structure
    \begin{equation}\label{eq33}
        {\beta _i} = {\left( {\gamma  - \frac{{\sigma _s^2}}{{T{\lambda _i}}}} \right)^ + }, \forall i,
    \end{equation}
where $\gamma > 0$ is chosen such that $\operatorname{tr}\left({\bm{\tilde R}}_X\right) = P_T$.
\end{proposition}
\begin{proof}
See Appendix E.
\end{proof}

\textit{Remark 3:} It is again noted from Proposition 3 that if the MSE lower bound is achieved by a unique sensing-optimal covariance matrix, the communication rate $I\left( {{{\mathbf{Y}}_c};{\mathbf{X}}\left| {{\mathbf{H}}_c} \right.} \right)$ will be reduced, as $\mathbf{R}_X =  T^{-1} \mathbf{X}\mathbf{X}^H$ is required to be deterministic. In this case, the ISAC waveform $\mathbf{X}$ is no longer Gaussian, where the only randomness (DoFs) lies in its right singular vectors, leading to the high-SNR capacity in Theorem 2.

\section{Conclusion}
In this paper, we investigated the fundamental limits of ISAC systems from a rate-distortion perspective. We considered a generic vector Gaussian ISAC channel model, and employed conditional MI as a performance metric for both S\&C, respectively. Our main results indicated that the sensing MI provides a universal lower bound for any well-defined distortion metrics of sensing, and that the distortion lower bound is minimized if the sample covariance matrix of the ISAC signal is deterministic. In such a case, the achievable communication rate is decreased due to the reduced randomness in the ISAC signal, leading to the DRT between S\&C. We also interpreted the main results by pointing out the analogy between wireless sensing and joint source-channel coding. That is, the sensing operation can be considered as the target, a non-cooperative information source, encodes and communicates the information of its parameters to the sensing Rx in a passive manner. Finally, we studied a specific example of target response matrix estimation for ISAC systems, and provided sufficient conditions for the achievability of the distortion lower bound. 
\balance{
\bibliographystyle{IEEEtran}
\bibliography{IEEEabrv,isac}
}

\clearpage

\appendices

\section{Proof of Lemma 1}
We first note that ${\boldsymbol{\upeta}} \to {{{\mathbf{H}}}_s} \to {{{\mathbf{Y}}}_s}$ forms a Markov chain. 
By applying the chain rule of the MI, we have
\begin{equation}\label{eq6}
\begin{gathered}
  I\left( {{{{\mathbf{Y}}}_s};{{{\mathbf{H}}}_s},{\bm\eta} \left| {\mathbf{X}} \right.} \right) = I\left( {{{{\mathbf{Y}}}_s};{\boldsymbol{\upeta}} \left| {\mathbf{X}} \right.} \right) + I\left( {{{{\mathbf{Y}}}_s};{{{\mathbf{H}}}_s}\left| {{\mathbf{X}},{\boldsymbol{\upeta}}} \right.} \right) \hfill \\
  \;\;\;\;\;\;\;\;\;\;\;\;\;\;\;\;\;\;\;\;\; \;\;\;\;\;= I\left( {{{{\mathbf{Y}}}_s};{{{\mathbf{  H}}}_s}\left| {\mathbf{X}} \right.} \right) + I\left( {{{{\mathbf{Y}}}_s};{\boldsymbol{\upeta}} \left| {{\mathbf{X}},{{{\mathbf{H}}}_s}} \right.} \right) \hfill \\
  \;\;\;\;\;\;\;\;\;\;\;\;\;\;\;\;\;\;\;\;\;\;\;\;\;\; \mathop   = \limits^{\left( a \right)}  I\left( {{{{\mathbf{Y}}}_s};{{{\mathbf{  H}}}_s}\left| {\mathbf{X}} \right.} \right),\hfill \\ 
\end{gathered}
\end{equation}
where $\left(a\right)$ in (\ref{eq6}) follows from the the Markov chain where ${{{\mathbf{Y}}}_s}$  is conditionally independent of $\boldsymbol{\upeta}$ given ${{{\mathbf{H}}}_s}$. Further, note that
\begin{equation}\label{eq7}
\begin{gathered}
  I\left( {{{{\mathbf{Y}}}_s};{{{\mathbf{H}}}_s}\left| {{\mathbf{X}},{\boldsymbol{\upeta}}} \right.} \right)  
  = 0, \hfill \\ 
\end{gathered} 
\end{equation}
which holds since $\mathbf{H}_s$ is a function of $\boldsymbol{\upeta}$. Combining (\ref{eq6}) and (\ref{eq7}) yields (\ref{eq5}) immediately, which completes the proof.

\section{Proof of Lemma 2}
We first vectorize the sensing signal matrix $\mathbf{Y}_s$ given $\mathbf{X}=\mathbf{A}$ as
\begin{equation}\label{eq8}
{{\mathbf{y}}_s} = \operatorname{vec} \left( {{{\mathbf{Y}}_s}} \right) = \left( {{{\bm{A}}^T} \otimes {{\mathbf{I}}_{{N_s}}}} \right){{\mathbf{h}}_s} + {{\mathbf{z}}_s} \triangleq  \bm{\tilde{A}}{{\mathbf{h}}_s} + {{\mathbf{z}}_s},
\end{equation}
where ${{\mathbf{h}}_s} = \operatorname{vec}\left({{\mathbf{H}}_s}\right)$, and ${{\mathbf{z}}_s} = \operatorname{vec}\left({{\mathbf{Z}}_s}\right)$. It is apparent that $I\left( {{{\mathbf{Y}}_s};\boldsymbol{\upeta} } \right) = I\left( {{{\mathbf{Y}}_s};{{\mathbf{H}}_s}} \right) = I\left( {{{\mathbf{y}}_s};{{\mathbf{h}}_s}} \right)$. Moreover, according to  \cite[Theorem 1]{MI_concave_complex}, $I\left( {{{\mathbf{y}}_s};{{\mathbf{h}}_s}} \right)$ is a concave function in $\bm{\tilde{A}}^H\bm{\tilde{A}}$. Note that
\begin{equation}\label{eq9}
{{\bm{\tilde{A}}}^H}\bm{\tilde{A}}= {\bm{{A}}^*}{\bm{{A}}^T} \otimes {{\mathbf{I}}_{{N_s}}} = T\bm{R}_A^* \otimes {{\mathbf{I}}_{{N_s}}}.
\end{equation}
Since the Kronecker product is a linear operator which preserves the concavity, $I\left( {{{\mathbf{Y}}_s};\boldsymbol{\upeta} } \right)$ is a concave function in $\bm{R}_{A}^{*}$, and is thus a concave function of $\bm{R}_A$. 

\section{Proof of Theorem 1}
By noting that ${\boldsymbol{\upeta}} \to {{{\mathbf{H}}}_s} \to ({{{\mathbf{Y}}}_s}, \mathbf{X}) \to {\hat{\boldsymbol{\upeta}}}$ forms a Markov chain, it holds that
\begin{equation}\label{eq12}
\begin{gathered}
  I\left( \boldsymbol{\upeta};\hat{\boldsymbol{\upeta}} \right)\mathop  \le\limits^{\left( a \right)} I\left( {{{\mathbf{H}}_s};{{\mathbf{Y}}_s},{\mathbf{X}}} \right) \hfill \\
  \;\;\;\;\;\;\;\;\;\;\;\;\; = I\left( {{{\mathbf{H}}_s};{{\mathbf{Y}}_s}\left| {\mathbf{X}} \right.} \right) + I\left( {{{\mathbf{H}}_s};{\mathbf{X}}} \right) \hfill \\
  \;\;\;\;\;\;\;\;\;\;\;\;\;\mathop  = \limits^{\left( b \right)} I\left( {{{\mathbf{H}}_s};{{\mathbf{Y}}_s}\left| {\mathbf{X}} \right.} \right) = I\left( {\boldsymbol{\upeta} ;{{\mathbf{Y}}_s}\left| {\mathbf{X}} \right.} \right), \hfill \\ 
\end{gathered}
\end{equation}
where $\left(a\right)$ is because of the data processing inequality, and $\left(b\right)$ is based on the fact that $\mathbf{H}_s$ and $\mathbf{X}$ are independent to each other (since the $\mathbf{H}_s$ is not available at the transmitter).

Let us consider the rate-distortion function of $p_{\boldsymbol{{\upeta}}}\left(\bm{\eta}\right)$, which is defined as
\begin{equation}\label{eq13}
R\left( D \right) = \mathop {\min }\limits_{p\left( {\hat{\boldsymbol{\upeta}} \left| \boldsymbol{\upeta}  \right.} \right)} \;I\left( \boldsymbol{\upeta};\hat{\boldsymbol{\upeta}} \right)\;\;\operatorname{s.t.}\;\mathbb{E}\left\{ {d\left(\boldsymbol{\upeta},\hat{\boldsymbol{\upeta}}\right)} \right\}  \le D.
\end{equation}
It follows that
\begin{equation}\label{eq14}
\begin{gathered}
R\left( D \right)  \le I\left( \boldsymbol{\upeta};\hat{\boldsymbol{\upeta}} \right) \le I\left( {\boldsymbol{\upeta} ;{{\mathbf{Y}}_s}\left| {\mathbf{X}} \right.} \right) \hfill \\ \quad \quad \quad =  {\mathbb{E}\left\{ {{I_{\boldsymbol{\upeta}} }\left( {{{\mathbf{R}}_X}} \right)} \right\}} \mathop  \le \limits^{\left( a \right)}   {{I_{\boldsymbol{\upeta}} }\left( \mathbb{E}\left\{{{{\mathbf{R}}_X}}\right\} \right)} = { {{I_{\boldsymbol{\upeta}} }\left( {{{{\bm{\tilde R}}}_X}} \right)}},
\end{gathered}
\end{equation}
where the Jensen's inequality $\left(a\right)$ holds since ${{I_{\boldsymbol{\upeta}} }\left( {{{\mathbf{R}}_X}} \right)}$ is a concave function in $\mathbf{R}_X$. The equality holds for $\left(a\right)$ if $\mathbf{R}_X$ satisfies the conditions given in Proposition 1.  

By noting the facts that $R\left(D\right)$ is a monotonic decreasing  function in $D$,   and that the distortion-rate function 
\begin{equation}\label{eq15}
D\left( R \right) = \mathop {\min }\limits_{p\left( {\hat{\boldsymbol{\upeta}} \left| \boldsymbol{\upeta}  \right.} \right)} \;\mathbb{E}\left\{ {d\left(\boldsymbol{\upeta},\hat{\boldsymbol{\upeta}}\right)} \right\}\;\;\operatorname{s.t.}\;I\left( \boldsymbol{\upeta};\hat{\boldsymbol{\upeta}} \right)  \le R
\end{equation}
is the inverse function of $R\left(D\right)$, (\ref{eq11}) holds immediately, which completes the proof.

\section{Proof of Proposition 2}
It is readily to see that the maximum of the sensing MI (\ref{eq18}) is reached if and only if ${\left|X \right|^2} = \mathbb{E}\left( {{{\left|X \right|}^2}} \right) = {P_T}$, i.e., $X$ has a constant amplitude $\sqrt{P_T}$. For the complex scalar Gaussian source  $H_s \sim \mathcal{CN}\left(0, \sigma^2_h\right)$, the corresponding rate-distortion and distortion-rate functions are
\begin{equation}\label{eq20}
    R_{G}\left( D \right) = \log \frac{{\sigma _h^2}}{D}, \quad D_{G}\left(R\right) = \sigma_h^2 2^{-R}.
\end{equation}
Now suppose that the sensing Rx applies an MMSE estimator to estimate $H_s$ for each realization of $X$. The MMSE for a given instance of $X$ can be expressed as
\begin{equation}\label{eq21}
\begin{gathered}
  \operatorname{mmse} \left( {{H_s}\left|X \right.} \right) = \frac{1}{{\sigma _h^{ - 2} + \sigma _s^{ - 2}{{\left| X \right|}^2}}}. \hfill \\ 
\end{gathered}
\end{equation}
The overall MMSE can be computed as
\begin{equation}\label{eq22}
\begin{gathered}
    \operatorname{mmse} \left( {{H_s}} \right) = \mathbb{E}\left\{ {\operatorname{mmse} \left( {{H_s}\left| X \right.} \right)} \right\} = \mathbb{E}\left\{ {\frac{1}{{\sigma _h^{ - 2} + \sigma _s^{ - 2}{{\left|X\right|}^2}}}} \right\} \hfill \\ 
\end{gathered}
\end{equation}
which is the minimum distortion achievable at the sensing Rx. It follows that
\begin{equation}\label{eq23}
\begin{gathered}
  R_{G}\left[ \operatorname{mmse} \left( {{H_s}} \right) \right] =  - \log \mathbb{E}\left\{ {{{\left( {1 + \frac{{\sigma _h^2{{\left| X\right|}^2}}}{{\sigma _s^2}}} \right)}^{ - 1}}} \right\} \hfill \\
  \mathop  \le\limits^{\left( a \right)}  - \mathbb{E}\log \left\{{{{\left( {1 + \frac{{\sigma _h^2{{\left|X \right|}^2}}}{{\sigma _s^2}}} \right)}^{ - 1}}} \right\} = I\left( {{Y_s};{H_s}\left|X\right.} \right), \hfill \\ 
\end{gathered}
\end{equation}
where $\left(a\right)$ is due to the fact that $ - \mathbb{E}\left\{ {\log x} \right\} \ge  - \log \mathbb{E}\left\{x\right\}$. Apparently, the equality holds for $\left(a\right)$ if and only if ${\left|X \right|^2} = \mathbb{E}\left\{ {{{\left| X \right|}^2}} \right\} = {P_T}$, in which case we have
\begin{equation}\label{eq24}
\begin{gathered}
  R_{G}\left[ \operatorname{mmse} \left( {{H_s}} \right) \right] = \log \left( {1 + \frac{{ P_T \sigma _h^2}}{{\sigma _s^2}}} \right) = I_{s,\max},\quad  \hfill \\
  \mathbb{E}\left\{ {{{\left| {{H_s} - {{\hat H}_s}} \right|}^2}} \right\} \ge \operatorname{mmse} \left( {{H_s}} \right) \hfill \\\;\; \quad\quad\quad\quad\quad\quad\quad = \frac{1}{{\sigma _h^{ - 2} + \sigma _s^{ - 2}{P_T}}}=\sigma _h^2{2^{ - {I_{s,\max }}}}. \hfill \\ 
\end{gathered}
\end{equation}
That is, the MSE attains the lower bound in Theorem 1 when the amplitude of $X$ is deterministic and fixed to $\sqrt{P_T}$.

\section{Proof of Proposition 3}
Following similar steps of (\ref{eq26}), the MMSE for estimating $\mathbf{H}_s$ can be obtained as
\begin{equation}\label{eq28}
\begin{gathered}
\begin{gathered}
  \operatorname{mmse}\left( {{{\mathbf{H}}_s}} \right) = \mathbb{E}\left\{ {\operatorname{mmse}\left( {{{\mathbf{H}}_s}\left| {\mathbf{X}} \right.} \right)} \right\} \hfill \\
   = \mathbb{E}\left\{ {{\text{tr}}\left[ {{{\left( {{\mathbf{\Lambda }}_h^{ - 1} + \sigma _s^{ - 2}T\sum\nolimits_{i = 1}^{{N_s}} {{{\mathbf{F}}_i}{{\mathbf{R}}_X}{\mathbf{F}}_i^H} } \right)}^{ - 1}}} \right]} \right\}. \hfill \\ 
\end{gathered}
\end{gathered}
\end{equation}

The rate-distortion function of the correlated vector Gaussian source $\mathbf{h}_s \sim \mathcal{CN}\left(\mathbf{0}, \bm{\Tilde{R}}_h\right)$ is equivalent to that of its independent counterpart ${{\mathbf{U}}^H}{{\mathbf{h}}_s} \sim \mathcal{CN}\left(\mathbf{0}, \bm{\Lambda}_h\right)$ \cite[Chapter 8]{cover}, which is
\begin{equation}\label{eq29}
    R_{VG}\left( D \right) = \sum\nolimits_{i = 1}^{{N_s}M} {\log \frac{{\lambda _i}}{{{{\left( {\mu  - \lambda _i} \right)}^ + } + \lambda _i}}},
\end{equation}
where $\lambda_i$ is the $i$th eigenvalue of $\bm{\Tilde{R}}_h$, and $\mu \ge 0$ is chosen such that $\sum\nolimits_{i = 1}^{{N_s}M} {{{\left( {\mu  - \lambda _i} \right)}^ + } + \lambda _i}  = D$. We further note that
\begin{equation}\label{eq30}
\begin{gathered}
  {R_{VG}}\left[ {\operatorname{mmse} \left( {{{\mathbf{H}}_s}} \right)} \right] = {R_{VG}}\left( {\mathbb{E}\left\{ {\operatorname{mmse} \left( {{{\mathbf{H}}_s}\left| {\mathbf{X}} \right.} \right)} \right\}} \right) \hfill \\
  \mathop  \le\limits^{\left( a \right)} \mathbb{E}\left\{ {{R_{VG}}\left[ {\operatorname{mmse} \left( {{{\mathbf{H}}_s}\left| {\mathbf{X}} \right.} \right)} \right]} \right\} \le \mathbb{E}\left\{ {I\left( {{{{\mathbf{\hat H}}}_s};{{\mathbf{H}}_s}\left| {{\mathbf{X}} = \bm{X}} \right.} \right)} \right\} \hfill \\
 \mathop  \le\limits^{\left( b \right)} \mathbb{E}\left\{ {I\left( {{{\mathbf{Y}}_s};{{\mathbf{H}}_s}\left| {{\mathbf{X}} = \bm{X}} \right.} \right)} \right\} = I\left( {{{\mathbf{Y}}_s};{{\mathbf{H}}_s}\left| {\mathbf{X}} \right.} \right), \hfill \\ 
\end{gathered}
\end{equation}
where the Jensen's inequality $\left(a\right)$ holds due to the fact that $R_{VG}\left( D \right)$ is a convex function in $D$, and $\left(b\right)$ is because of the data processing inequality for the Markov chain ${{{\mathbf{H}}}_s} \to \left({{{\mathbf{Y}}}_s}, \mathbf{X}\right) \to {\hat{\mathbf{H}}}_s$ for a given $\mathbf{X}$. Following Theorem 1, we have
\begin{equation}\label{eq31}
    \operatorname{mmse} \left( {{{\mathbf{H}}_s}} \right) \ge {D_{VG}}\left[ {I\left( {{{\mathbf{Y}}_s};{{\mathbf{H}}_s}\left| {\mathbf{X}} \right.} \right)} \right],
\end{equation}
where $D_{VG}\left(R\right)$ is the distortion-rate function for independent vector Gaussian sources.

Note that (\ref{eq27}) and (\ref{eq28}) are concave and convex functions of $\mathbf{R}_X$, respectively. By the Jensen's inequality we have
\begin{equation}\label{eq34}
\begin{gathered}
  I\left( {{{\mathbf{Y}}_s};{{\mathbf{H}}_s}\left| {\mathbf{X}} \right.} \right) \hfill \\
   = \mathbb{E}\left\{ {\log \left| {{\mathbf{I}} + \sigma _s^{ - 2}T{{\mathbf{\Lambda }}_h}\sum\nolimits_{i = 1}^{{N_s}} {{{\mathbf{F}}_i}{{\mathbf{R}}_X}{\mathbf{F}}_i^H} } \right|} \right\} \hfill \\
   \mathop  \le\limits^{\left( a \right)} \log\left| {{\mathbf{I}} + \sigma _s^{ - 2}T{{\mathbf{\Lambda }}_h}\sum\nolimits_{i = 1}^{{N_s}} {{{\mathbf{F}}_i}\mathbb{E}\left( {{{\mathbf{R}}_X}} \right){\mathbf{F}}_i^H} } \right| \hfill \\
   \mathop  \le\limits^{\left( b \right)} \mathop {\max }\limits_{\begin{subarray}{c} 
  {{{\bm{\tilde R}}}_X} \succeq \mathbf{0},\; \\ 
  \operatorname{tr} \left( {{{{\bm{\tilde R}}}_X}} \right)\; = {P_T} 
\end{subarray}}  \log\left| {{\mathbf{I}} + \sigma _s^{ - 2}T{{\mathbf{\Lambda }}_h}\sum\nolimits_{i = 1}^{{N_s}} {{{\mathbf{F}}_i}{{{\bm{\tilde R}}}_X}{\mathbf{F}}_i^H} } \right|,\hfill \\ 
\end{gathered} 
\end{equation}
and 
\begin{equation}\label{eq35}
\begin{gathered}
  \operatorname{mmse} \left( {{{\mathbf{H}}_s}} \right) \hfill \\
   = \mathbb{E}\left\{ {{\text{tr}}\left[ {{{\left( {{\mathbf{\Lambda }}_h^{ - 1} + \sigma _s^{ - 2}T\sum\nolimits_{i = 1}^{{N_s}} {{{\mathbf{F}}_i}{{\mathbf{R}}_X}{\mathbf{F}}_i^H} } \right)}^{ - 1}}} \right]} \right\} \hfill \\
   \mathop  \ge\limits^{\left( c \right)} {\text{tr}}\left[ {{{\left( {{\mathbf{\Lambda }}_h^{ - 1} + \sigma _s^{ - 2}T\sum\nolimits_{i = 1}^{{N_s}} {{{\mathbf{F}}_i}\mathbb{E}\left( {{{\mathbf{R}}_X}} \right){\mathbf{F}}_i^H} } \right)}^{ - 1}}} \right] \hfill \\
   \mathop  \ge\limits^{\left( d \right)} \mathop {\min }\limits_{\begin{subarray}{c} 
  {{{\bm{\tilde R}}}_X} \succeq \mathbf{0},\; \\ 
  \operatorname{tr} \left( {{{{\bm{\tilde R}}}_X}} \right)\; = {P_T} 
\end{subarray}}  {\text{tr}}\left[ {{{\left( {{\mathbf{\Lambda }}_h^{ - 1} + \sigma _s^{ - 2}T\sum\nolimits_{i = 1}^{{N_s}} {{{\mathbf{F}}_i}{{{\bm{\tilde R}}}_X}{\mathbf{F}}_i^H} } \right)}^{ - 1}}} \right], \hfill \\ 
\end{gathered} 
\end{equation}
where the equal signs hold for both inequalities $(a)$ and $(c)$ if condition 2 holds, i.e., ${{\mathbf{R}}_X} = \mathbb{E}\left( {{{\mathbf{R}}_X}} \right) = {\bm{\tilde R}}_X$.

Moreover, by noting the fact that
\begin{equation}\label{eq36}
\begin{gathered}
  {{{\mathbf{\tilde F}}}^T}{{{\mathbf{\tilde F}}}^*} = \left[ \begin{gathered}
  {\mathbf{F}}_1^H \hfill \\
  {\mathbf{F}}_2^H \hfill \\
   \vdots  \hfill \\
  {\mathbf{F}}_{{N_s}}^H \hfill \\ 
\end{gathered}  \right]\left[ {{{\mathbf{F}}_1},{{\mathbf{F}}_2}, \ldots ,{{\mathbf{F}}_{{N_s}}}} \right] = {{\mathbf{K}}^T}{{\mathbf{U}}^*}{{\mathbf{U}}^T}{\mathbf{K}} = {{\mathbf{I}}_{M{N_s}}}, \hfill \\ 
\end{gathered}
\end{equation}
we have ${\mathbf{F}}_i^H{{\mathbf{F}}_i} = {{\mathbf{I}}_M}, \forall i$. It follows that
\begin{equation}\label{eq37}
    \operatorname{tr} \left( {\sum\nolimits_{i = 1}^{{N_s}} {{{\mathbf{F}}_i}{{{\bm{\tilde R}}}_X}{\mathbf{F}}_i^H} } \right) = \sum\nolimits_{i = 1}^{{N_s}} {\operatorname{tr} \left( {{{{\bm{\tilde R}}}_X}{\mathbf{F}}_i^H{{\mathbf{F}}_i}} \right)}  = {N_s}{P_T}.
\end{equation}

Now suppose that both conditions 2 and 3 hold. The sensing MI and MMSE can be simplified as
\begin{equation}\label{eq38}
\begin{gathered}
  I\left( {{{\mathbf{Y}}_s};{{\mathbf{H}}_s}\left| {\mathbf{X}} \right.} \right) = \log \left| {{\mathbf{I}} + \sigma _s^{ - 2}T{{\mathbf{\Lambda }}_h}\sum\nolimits_{i = 1}^{{N_s}} {{{\mathbf{F}}_i}{{{\bm{\tilde R}}}_X}{\mathbf{F}}_i^H} } \right| \hfill \\
 \quad\quad\quad\quad\quad\quad\; = \sum\nolimits_{i = 1}^{{N_sM}} {\log \left( {1 + \sigma _s^{ - 2}T{\lambda _i}{\beta _i}} \right)} ,  \hfill \\ 
\end{gathered}
\end{equation}
\begin{equation}\label{eq39}
\begin{gathered}
  \operatorname{mmse} \left( {{{\mathbf{H}}_s}} \right) = {\text{tr}}\left[ {{{\left( {{\mathbf{\Lambda }}_h^{ - 1} + \sigma _s^{ - 2}T\sum\nolimits_{i = 1}^{{N_s}} {{{\mathbf{F}}_i}{{{\bm{\tilde R}}}_X}{\mathbf{F}}_i^H} } \right)}^{ - 1}}} \right] \hfill \\
   \;\;\;\;\;\;\;\;\;\;\;\;\;\;\;\;\;\;= \sum\nolimits_{i = 1}^{{N_sM}} {\frac{{{\lambda _i}}}{{1 + \sigma _s^{ - 2}T{\lambda _i}{\beta _i}}}},\hfill \\ 
\end{gathered}
\end{equation}
where $\sum\nolimits_{i = 1}^{{N_s}} {{\beta _i} = {N_s}{P_T}}$. It can be readily verified that the water-filling solution (\ref{eq33}) achieves the maximum and minimum of (\ref{eq38}) and (\ref{eq39}) for a given power budget $P_T$, respectivel. Accordingly, the optimal sensing MI and MMSE can be recast as
\begin{equation}\label{eq40}
  I\left( {{{\mathbf{Y}}_s};{{\mathbf{H}}_s}\left| {\mathbf{X}} \right.} \right) = \sum\nolimits_{i = 1}^{{N_s}} {\log } \left( {1 + {{\left( {\sigma _s^{ - 2}T{\lambda _i}\gamma  - 1} \right)}^ + }} \right),
\end{equation}
\begin{equation}\label{eq41}
    \operatorname{mmse} \left( {{{\mathbf{H}}_s}} \right) = \sum\nolimits_{i = 1}^{{N_s}} {\frac{{{\lambda _i}}}{{1 + {{\left( {\sigma _s^{ - 2}T{\lambda _i}\gamma  - 1} \right)}^ + }}}}.
\end{equation}
Upon recalling (\ref{eq29}) and letting $\mu = \sigma_s^2T^{-1}\gamma^{-1}$, we have
\begin{equation}\label{eq42}
\sum\nolimits_{i = 1}^{{N_s}M} {{{\left( {\mu  - {\lambda _i}} \right)}^ + } + } {\lambda _i} = \operatorname{mmse} \left( {{{\mathbf{H}}_s}} \right),
\end{equation}
\begin{equation}\label{eq43}
\begin{gathered}
  {R_{VG}}\left[ {\operatorname{mmse} \left( {{{\mathbf{H}}_s}} \right)} \right] = \sum\nolimits_{i = 1}^{{N_s}M} {\log \frac{{{\lambda _i}}}{{{{\left( {\mu  - {\lambda _i}} \right)}^ + } + {\lambda _i}}}}  \hfill \\
  \;\;\;\;\;\;\;\;\;\;\;\;\;\;\;\;\;\;\;\;\;\;\;\;\;\;\;\;\; = I\left( {{{\mathbf{Y}}_s};{{\mathbf{H}}_s}\left| {\mathbf{X}} \right.} \right), \hfill \\ 
\end{gathered}
\end{equation}
which implies $\operatorname{mmse} \left( {{{\mathbf{H}}_s}} \right) = {D_{VG}}\left[ {I\left( {{{\mathbf{Y}}_s};{{\mathbf{H}}_s}\left| {\mathbf{X}} \right.} \right)} \right]$, completing the proof.
\end{document}